\documentclass[12pt]{iopart}

\expandafter\let\csname equation*\endcsname\relax
\expandafter\let\csname endequation*\endcsname\relax

\usepackage{iopams,amsmath,amsthm,mathbbol,url,graphicx,enumerate}
\usepackage{color}

\newtheorem{thm}{Theorem}[section]

\newtheorem{cor}[thm]{Corollary}
\newtheorem{lem}[thm]{Lemma}
\newtheorem{prop}[thm]{Proposition}

\newtheorem{rem}[thm]{Remark}

\newcommand{\partd}[2][]{\frac{\partial #1}{\partial #2}}

\newcommand{\st}{\mathop{} \big| \mathop{}}
\newcommand{\set}[1]{\left\{ #1 \right\}}

\begin{document}

\title{Generic Ising trees}

\author{B. Durhuus$^1$, G. M. Napolitano$^2$}
\address{Dept. of Mathematical Sciences, University of Copenhagen, Universitetsparken 5, 2100 Copenhagen \O, Denmark}
\ead{$^1$durhuus@math.ku.dk, $^2$gmn@math.ku.dk}

\begin{abstract}
The Ising model on a class of infinite random trees is defined as a
thermodynamic limit of finite systems. A detailed description of the
corresponding distribution of infinite spin configurations is
given. As an application we study the magnetization properties of such
systems and prove that they exhibit no spontaneous
magnetization. Furthermore, the values of the Hausdorff and spectral
dimensions of the underlying trees are calculated and found to be,
respectively, $\bar{d}_h=2$ and $\bar{d}_s=4/3$. 
\end{abstract}


\maketitle

\section{Introduction}
Since its appearance, the Ising model has been considered in various
geometrical backgrounds. Most familiar are the regular lattices, where
it is well known that in dimension $d=1$, originally considered by
Ising and Lenz \cite{Ising,Lenz}, there is no phase transition as
opposed to dimension $d \geq 2$, where spontaneous magnetization
occurs at sufficiently low temperature \cite{Onsager, Stanley}. 

The Ising model on a Cayley tree turns out to be exactly solvable
\cite{Domb, Egg, Katz-Takiz, M-HZ}. Despite the fact that the free energy, in this
case, is an analytic function of the temperature at vanishing magnetic
field, the model does have a phase transition and exhibits spontaneous
magnetization in the central region of the tree, also called Bethe
lattice \cite{Baxter}. One may attribute this unusual behavior to the
large size of the boundary of a ball in the tree as compared to its
volume. This result has been generalized to non-homogeneous graphs in \cite{Lyons} (see also references therein).

Studies of the Ising model on non-regular graphs are generally
non-tractable from an analytic point of view. For numerical studies
see e.g. \cite{BJK}. See also \cite{BL}, where the Ising model with external field
coupled to the causal dynamical triangulation model is studied via high- and low-temperature expansion
techniques. In \cite{ADJ} a grand canonical ensemble of Ising
models on random finite trees was considered, motivated by studies in
two dimensional quantum gravity \cite{ADJ-qg}. It was argued in
\cite{ADJ} that the model does not exhibit spontaneous magnetization
at values of the fugacity where the mean size of the trees diverges. 

In the present paper we study the Ising model with an external magnetic field on certain infinite
random trees, constructed as ``thermodynamic'' limits of Ising systems
on random finite trees. The models are defined in terms of branching weights $p_n$, $n \geq 0$, associated to
vertices and subject to a certain \emph{genericity condition} \cite{DJW2}. The latter is a rather mild
requirement, satisfied e.g. by all non-linear trees of bounded vertex degree, i.e. if $p_n = 0$ for
$n$ large enough, and more generally, if the generating function $\sum_n p_n t^n$ has infinite
radius of convergence. Another example is the so called uniform infinite tree corresonding to
$p_n = 1$, $n \geq 0$. The precise form of the condition is stated in eq. (\ref{eq:genass}) and
used as an important ingredient in the construction of the infinite-size limit that
will be referred to as a \emph{generic Ising tree}.

Using tools developed in \cite{BD03, DJW2} we prove for such
ensembles that spontaneous magnetization is absent. The basic reason
is that the generic infinite tree has a certain one dimensional
feature despite the fact that we prove its Hausdorff dimension to be
2. Furthermore, we obtain results on the spectral dimension of
generic Ising trees.

We remark that if the genericity condition is violated the random trees
under consideration tend to develop vertices of infinite order in the
limit of infinite size \cite{JS09, JS11}, and will not be treated in this paper.

This article is organized as follows. After a brief review of some
basic graph theoretic notions that will be used throughout the article
and fixing some notation we define, in Section \ref{sec:gen-set}, the
finite size systems whose infinite size limits are our main object of
study. The remainder of Section \ref{sec:gen-set} is devoted to an overview of
the main results, including the existence and detailed description
of the infinite size limit, the magnetization properties and the
determination of the annealed Hausdorff and spectral dimensions of
generic Ising trees.

The next two sections provide detailed proofs and, in some
cases, more precise statements of those results. Under the genericity
assumption mentioned above we determine, in Section \ref{sec:inf-measure}, the asymptotic behavior of
the partition functions of ensembles of spin systems on finite trees
of large size. This allows a construction of the limiting distribution on
infinite trees and also leads to a precise description of the limit. 
In Section  \ref{sec:dim} we exploit the
latter characterization to determine the annealed Hausdorff and spectral
dimensions of the generic Ising trees, whereafter we establish absence
of magnetization in Section \ref{sec:mag-prop}. 

Finally, some concluding remarks on possible future developments are collected in Section \ref{sec:concl}.

\section{Definition of the models and main results}
\label{sec:gen-set}

\subsection{Basic definitions}
Recall that a \emph{graph} $G$ is specified by its \emph{vertex} set $V(G)$ and
its \emph{edge} set $E(G)$. Vertices will be denoted by $v$ or
$v_i$ etc. An edge is then an unordered pair $(v,v')$ of different
vertices. A vertex $v'$ is called a \emph{neighbor} of $v$, if $(v,v') \in E(G)$. 
Both finite and infinite graphs will be considered,
i.e. $V(G)$ may be finite or infinite, and all graphs will be assumed
to be locally finite, i.e. the number $\sigma_v$ of edges containing a
vertex $v$, called the \emph{degree} (or \emph{order}) of $v$, is finite for all $v\in
V(G)$. By the \emph{size} of $G$ we shall mean the number of edges in $G$
and denote it by $|G|$, i.e. $|G|= \sharp E(G)$, where $\sharp M$ is
used to denote the number of elements in a set $M$.

A \emph{path} in $G$ is a sequence of different edges
$(v_0,v_1),(v_1,v_2),\ldots,(v_{k-1},v_k)$ where $v_0$ and $v_k$ are
called the \emph{end vertices}. If $v_0=v_k$ the path is called a
\emph{circuit} originating at $v_0$. The graph $G$ is called
\emph{connected} if any two vertices $v$ and $v'$ of $G$ can be
connected by a path, i.\,e. they are end vertices of a path. The
\emph{graph distance} between $v$ and $v'$ is then defined as the
minimal number of edges in a path connecting them. A connected graph
is called a \emph{tree} if it has no circuits. 

Given a connected graph $G$, $R\geq 0$ and $v \in V(G)$, we denote by
$B_R(G,v)$ the \emph{closed ball} of radius $R$ centered at $v$, i.e. 
$B_R(G,v)$ is the subgraph of $G$ spanned by the vertices at graph distance 
$\leq R$ from $v$.

A \emph{rooted} tree is a tree with a distinguished vertex $r$ called the \emph{root vertex}, which will be assumed to be of order 1 in the following. For a rooted tree, we define the \emph{children} of a vertex $v$, at distance $k$ from the root, as the neighboring vertices of $v$ at distance $k+1$ from $r$. A \emph{plane} tree is a rooted tree together with an ordering of the children of each vertex\footnote{The terms \emph{ordered} tree and  \emph{planar} tree are also used in the literature.}. 

We denote by $\mathcal{T}$ the set of such trees, 
by $\mathcal{T}_N$ the subset of $\mathcal{T}$ of trees of size $N$
and by $\mathcal{T}_\infty$ the subset of infinite trees, such that
\begin{equation}
	\mathcal{T} = \left( \bigcup_{N=1}^\infty \mathcal{T}_N \right) \cup \mathcal{T}_\infty.
\end{equation}

The \emph{height} of a finite tree is the maximal distance from the root to one of its vertices.

The set $\mathcal{T}$ is a metric space with the distance
between two trees $\tau$ and $\tau'$ defined by 
\begin{equation}
	\tilde{d}(\tau,\tau') = \inf \set{\frac{1}{R+1} \st B_R(\tau)=B_R(\tau')},
\label{eq:metric-T}
\end{equation}
where $B_R(\tau)$ denotes the ball of radius $R$ centered at the root $r$, i.e. $B_R(\tau) \equiv B_R(\tau,r)$. See \cite{BD03} for further details on properties of $\tilde{d}$. In particular, $\tilde{d}$ is an ultrametric, i.e.
\begin{equation}
	\tilde{d}(\tau,\tau') \leq \max\set{\tilde{d}(\tau,\tau''),\tilde{d}(\tau',\tau'')}
\end{equation}
for all $\tau$, $\tau''$, $\tau'' \in \mathcal{T}$.

\subsection{The models and the thermodynamic limit}

The statistical mechanical models considered in this paper are defined
in terms of plane trees as follows. Let $\Lambda_N$ be the set of rooted
plane trees of size $N$ decorated with Ising spin configurations,
\begin{equation}
	\Lambda_N = \set{ s \, : \, V(\tau) \to \set{\pm 1} \st \tau \in \mathcal{T}_N},
\end{equation}
and set
\begin{equation}
	\Lambda = \left( \bigcup_{N=1}^\infty \Lambda_N \right) \cup \Lambda_\infty,
\end{equation}
where $\Lambda_\infty$ denotes the set of infinite decorated trees. In the following we will often denote by $\tau_s$ a generic element of $\Lambda$, in particular when stressing the underlying tree structure $\tau$ of the spin configuration $s$. Furthermore, we shall use both $s_v$ and $s(v)$ to denote the value of the spin at vertex $v$.

The set $\Lambda$ is a metric space with metric $d$ defined by
\begin{equation}
	d(\tau_{s},\tau'_{s'}) = \inf\set{\frac{1}{R+1} \st B_R(\tau) = B_R(\tau'), \, s|_{B_R(\tau)}=s'|_{B_R(\tau')}},
\label{eq:metric-Lambda}
\end{equation}
as a generalization of (\ref{eq:metric-T}).

We define a probability measure $\mu_N$ on $\Lambda_N$ by
\begin{equation}\label{muN}
	\mu_N(\tau_s) = \frac{1}{Z_N} e^{-H(\tau_s)} \rho(\tau), \quad \tau_s \in \Lambda_N,
\end{equation}
where the Hamiltonian $H(\tau_s)$, describing the interaction of each spin with 
its neighbors and with the constant external magnetic field $h$ at inverse
temperature $\beta$, is given by
\begin{equation}
	H(\tau_s) = - \beta \sum_{(v_i,v_j)\in E(\tau)} s_{v_i} s_{v_j} - h \sum_{v_i \in V(\tau)\setminus r} s_{v_i}.
\end{equation}
The weight function $\rho(\tau)$ is defined in terms of the \textit{branching weights} $p_{\sigma_v - 1}$ associated with vertices $v \in V(\tau) \setminus r$, and is given by
\begin{equation}
	\rho(\tau) = \prod_{v \in V(\tau) \setminus r} p_{\sigma_v - 1}.
\end{equation}
Here $(p_n)_{n\geq0}$ is a sequence of non-negative numbers such that
$p_0 \neq 0$ and $p_n\neq 0$ for some $n \geq 2$ (otherwise only
linear chains would contribute). We will further assume the branching
weights to satisfy a \emph{genericity condition} explained below in
(\ref{eq:genass}), and which defines the \emph{generic Ising tree
  ensembles} considered in this paper (see also \cite{DJW2}). Finally,
the partition function $Z_N$ in \eqref{muN} is given by
\begin{equation}
	Z_N(\beta,h) = \sum_{\tau \in \mathcal{T}_N} \sum_{s \in S_\tau} e^{-H(\tau_s)} \rho(\tau),
\label{eq:part-func-N}
\end{equation}
where $S_\tau = \set{\pm 1}^{V(\tau)}$.

We note that in the following the measure $\mu_N$ will be considered as a measure on $\Lambda$ supported on the finite set $\Lambda_N$, that is
\begin{equation}
	\mu_N(\Lambda \setminus \Lambda_N) = 0.
\end{equation}

Our first result (see Sec. \ref{sec:inf-measure}) establishes the existence 
of the thermodynamic limit of this model, in the sense that we prove the existence of a 
limiting probability measure $\mu = \mu^{(\beta,h)}=\lim_{N \to \infty} \mu_N$ defined on the set of trees of
infinite size decorated with spin configurations. Here, the limit 
should be understood in the weak sense, that is
\begin{equation}
	\int_\Lambda f(\tau_s) \, d \mu_N(\tau_s) \xrightarrow{N\rightarrow \infty} \int_\Lambda f(\tau_s) \, d \mu(\tau_s)
\end{equation}
for all bounded continuous functions $f$ on $\Lambda$. In particular, we find that the measure 
$\mu$ is concentrated on the set of infinite trees with a single infinite path, the \emph{spine}, starting at the root $r$, 
and with finite trees attached to the spine vertices, the \emph{branches}, see Fig. \ref{fig:grt}.

\begin{figure}%
\centering
\includegraphics[width=\columnwidth]{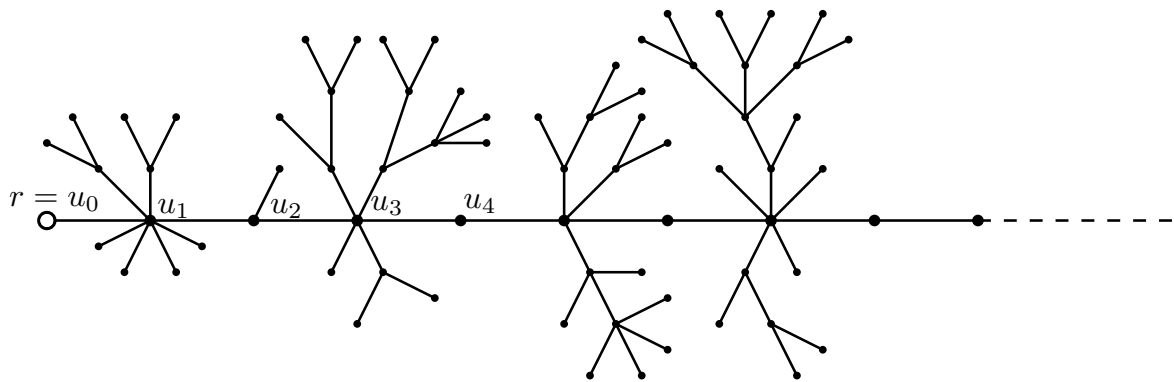}%
\caption{Example of an infinite tree, consisting of a spine and left and right branches.}%
\label{fig:grt}%
\end{figure}

As will be shown, the limiting distribution $\mu$ can be expressed in
explicit terms in such a way that a number of its 
characteristics, such as the Hausdorff dimension, the spectral dimension, as well as the
magnetization properties of the spins, can be analyzed in some
detail. For the reader's convenience we now give a brief account of those results.

\subsection{Magnetization properties}
\label{sec:sm-mag}
As a first result we show that the generic Ising tree exhibits no single site spontaneous magnetization 
at the root $r$ or at any other spine vertex, i.\,e. 
\begin{equation}
	\lim_{h \to 0} \mu^{(\beta,h)}(\set{\tau_s \st s(v) = + 1}) = \frac{1}{2}\,,
\end{equation}
for any vertex $v$ on the spine and all $\beta \in \mathbb{R}$. Details of this result can be found in Theorem \ref{thm:nomag-spine}.

The fact that the measure $\mu$ is supported on trees with a single
spine  gives rise to an analogy with the one-dimensional Ising model. In fact, we
show that the spin distribution on the spine equals that of the Ising
model on the half-line at the same temperature but in a modified
external magnetic field. As a consequence, we find that also the mean magnetization of the spine
vanishes for $h \to 0$. 

A different and perhaps more relevant result concerns the the total mean
magnetization, which may be stated as follows. First, let us 
define the mean magnetization in the ball of radius $R$ around the root by 
\begin{equation}\label{mag1}
	M_R(\beta,h) = \left< |B_R(\tau)| \right>^{-1}_{\beta,h} \left< \sum_{v \in B_R(\tau)} s_v \right>_{\beta,h}
\end{equation}
and the mean magnetization on the full infinite tree as
\begin{equation}\label{mag2}
	M(\beta,h) = \limsup_{R \rightarrow \infty} M_R(\beta,h).
\end{equation}
Here $\langle \cdot \rangle_{\beta,h}$ denotes the expectation value w.r.t. $\mu^{(\beta,h)}$. For the generic Ising tree, we prove in Theorem \ref{thm:no-mean-mag}  that this quantity satisfies
\begin{equation}
	\lim_{h \rightarrow 0} M(\beta,h) = 0, \quad \beta \in \mathbb{R}\,.
\end{equation}

It should be noted that for a fixed infinite homogeneous lattice, such as the Bethe lattice mentioned previously, the single site magnetization is constant over the lattice and hence equals the mean magnetization as defined above. The systems under consideration in this paper are defined on random lattices and do not possess any obvious homogeneity properties.

\subsection{Hausdorff dimension}
Given an infinite connected graph $G$, if the limit 
\begin{equation}\label{Haus1}
d_h = \lim_{R\to\infty}\frac{\ln |B_R(G,v)|}{\ln R}
\end{equation}
exists, we call $d_h$ the \emph{Hausdorff dimension} of $G$. It is easily seen that the
existence of the limit as well as its value do not depend on the
vertex $v$.

For an ensemble of infinite graphs $\mathcal{G}_\infty$ with a probability measure $\nu$,
we define the \emph{annealed} Hausdorff dimension by
\begin{equation}\label{Haus2}
\bar d_h = \lim_{R\to\infty}\frac{\ln\,
  \langle\,|B_R(G)|\,\rangle_{\nu}}{\ln R}\,, 
\end{equation}
provided the limit exists, where $<\cdot>_\nu$ denotes the expectation
value w.r.t. $\nu$.

We show in Theorem \ref{thm:Haus-dim} that the annealed Hausdorff
dimension of a generic Ising tree can be evaluated and equals that of
generic random trees as introduced in \cite{DJW2}, i.e. 
\begin{equation}
	\bar{d}_h = 2\,.
\end{equation}

\subsection{Spectral dimension}\label{sec:spec-dim}
A \emph{walk} on a graph $G$ is a sequence
$(v_0,v_1),(v_1,v_2),\ldots,(v_{k-1},v_k)$ of (not necessarily
different) edges in $G$. We shall denote such a walk by $\omega:
v_0\to v_k$ and call $v_0$ the \emph{origin} and $v_k$ the \emph{end}
of the walk. Moreover, the number $k$ of edges in $\omega$ will be
denoted by $|\omega|$. To each such walk $\omega$ we associate a
weight
\begin{equation}
	\pi_G(\omega) = \prod_{i=0}^{|\omega|-1} \sigma_{\omega(i)}^{-1}
\end{equation}
where $\omega(i)$ is the i'th vertex in $\omega$. Denoting by
$\Pi_n(G,v_0)$ the set of walks of length $n$ originating at vertex $v_0$   
we have
\begin{equation}
	\sum_{\omega\in\Pi_n(G,v_0)} \pi_G(\omega) = 1\,.
\end{equation}
i.e. $\pi_G$ defines a probability distribution on $\Pi_n(G,v_0)$. We call
$\pi_G$ the \emph{simple random walk} on $G$.

For an infinite connected graph $G$ and $v\in V(G)$ we denote by $\pi_t(G,v)$ the
\emph{return probability} of the simple random walk to $v$ at time
$t$, that is
\begin{equation}
	\pi_t(G,v) = \sum_{\substack{\omega:v\to v \\ |\omega|=t}} \pi_G(\omega)\,. 
\end{equation}
One can in a standard manner relate this quantity to the discrete heat
kernel on $G$, but we shall not need this interpretation in the following.
If the limit 
\begin{equation}\label{spectral1}
d_s = - 2\,\lim_{t\to\infty}\frac{\ln \pi_t(G,v)}{\ln t}
\end{equation}
exists, we call $d_s$ the \emph{spectral dimension} of $G$. Again in
this case, the existence and value of the limit are
independent of $v$. 

If $G$ is the hyper-cubic lattice $\mathbb Z^d$ it is clear that
$d_h=d$ and by Fourier analysis it is straight-forward to see that
also $d_s=d$.  
However, examples of graphs with $d_h\neq d_s$ are abundant, see
e.g. \cite{DJW1}.
 
The \emph{annealed} spectral dimension of an ensemble 
$(\mathcal{G}_\infty,\nu)$ of rooted infinite graphs is defined as
\begin{equation}\label{spectral2}
\bar d_s = - 2\,\lim_{t\to\infty}\frac{\ln\;\langle\,
  \pi_t(G,r)\,\rangle_{\nu}}{\ln t} 
\end{equation}
provided the limit exists.

We show in Theorem \ref{thm:spect-dim} that the annealed spectral dimension  
of a generic Ising tree is
\begin{equation}
	\bar{d}_s = \frac{4}{3}.
\end{equation}

The values of the Hausdorff dimension and the spectral dimension of generic Ising trees 
are thus found to coincide with those of generic random trees \cite{DJW2}. This indicates that the geometric structure of the
underlying trees is not significantly influenced by the coupling to
the Ising model as long as the model is generic.

\section{Ensembles of infinite trees}
\label{sec:inf-measure}

In this section we establish the existence of the measure $\mu^{(\beta,h)}$ on the set of infinite trees for values of $\beta, h$ that will be specified below. Our starting point is the Ising model on finite but large trees. We first consider the dependence of its partition function on the size of trees.   

\subsection{Asymptotic behavior of partition functions}
\label{sec:part-func}
 
Let the branching weights $(p_n)_{n \geq 0}$ be given as above and consider the generating functions
\begin{equation}
	\varphi(z) = \sum_{n=0}^\infty p_n z^n,
\label{eq:generat_weight}
\end{equation}
which we assume to have radius of convergence $\xi > 0$, and
\begin{equation}
	Z(\beta,h,g) = \sum_{N=0}^\infty Z_N(\beta,h) g^N,
\end{equation} 
where $Z_N$ is given by (\ref{eq:part-func-N}).

Decomposing the set $S_\tau$ into the two disjoint sets
\begin{equation}
	S_\tau^{\pm} = \set{s \in S_\tau  \st s(r) = \pm 1},
\end{equation}
gives rise to the decompositions
\begin{equation}
	\Lambda_N = \Lambda_{N+} \cup \Lambda_{N-}
\end{equation}
and
\begin{equation}
	\Lambda = \Lambda_+ \cup \Lambda_-.
\end{equation}
Correspondingly, we get
\begin{equation}
	Z(\beta,h,g) = Z_+(\beta,h,g) + Z_-(\beta,h,g),
\end{equation}
where the generating functions $Z_{\pm}(\beta,h,g)$ are given by 
\begin{equation}
	Z_\pm(\beta,h,g) = \sum_{N=0}^\infty Z_{N\pm}(\beta,h) g^N,
\end{equation}
and $Z_{N\pm}$ are defined by restricting the second sum in (\ref{eq:part-func-N}) to  $S_\tau^{\pm}$.

\begin{figure}%
\centering
\includegraphics[width=\columnwidth]{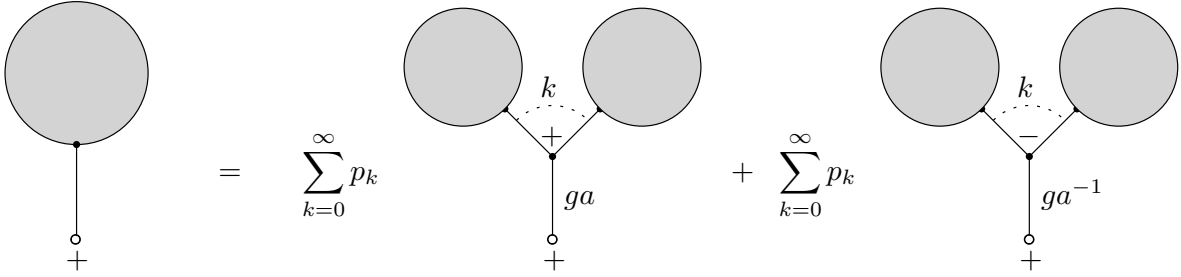}%
\caption{Tree decomposition corresponding to the first equation of the system (\ref{eq:zpzm}). The
  tree is decomposed according to the spin and the degree of the
  root's neighbor.}%
\label{fig:treedecomp}%
\end{figure}

Decomposing the tree as in Fig.\ref{fig:treedecomp}, it is easy to see that the functions $Z_\pm(g)$ are determined by the system of equations
\begin{equation}
	\begin{cases}
		Z_+ = g (a \, \varphi(Z_+) + a^{-1} \, \varphi(Z_-))\\
		Z_- = g (b \, \varphi(Z_+) + b^{-1} \, \varphi(Z_-))
	\end{cases}
\label{eq:zpzm}
\end{equation}
where
\begin{equation}
	a = e^{\beta + h}, \quad b = e^{-\beta + h}. 
\end{equation}

Let us define $F : \set{ |z|< \xi }^2 \times \mathbb{C} \rightarrow \mathbb{C}^2$ by
\begin{equation}
	F(Z_+,Z_-,g) = \mathcal{Z} - g \, \Phi(Z_+,Z_-),
\label{eq:f}
\end{equation}
where,
\begin{equation}
	\mathcal{Z} \equiv 
	\begin{pmatrix}
			Z_+ \\
			Z_-
		\end{pmatrix},
	\qquad \Phi(Z_+,Z_-) \equiv
	\begin{pmatrix}
			a \, \varphi(Z_+) + a^{-1} \, \varphi(Z_-)\\
			b \, \varphi(Z_+) + b^{-1} \, \varphi(Z_-)
		\end{pmatrix}.
\end{equation}
With the assumption $\xi >0$, we have
\begin{equation}
	\partd[F]{\mathcal{Z}} = \mathbb{1} - g \, \partd[\Phi]{\mathcal{Z}} = \mathbb{1} - g
		\begin{pmatrix}
			a \, \varphi'(Z_+) & a^{-1} \, \varphi'(Z_-)\\
			b \, \varphi'(Z_+) & b^{-1} \, \varphi'(Z_-)
		\end{pmatrix},
\end{equation}
and in particular,  $F(0,0,0)=0$ and $\partd[F]{\mathcal{Z}}(0,0,0) = \mathbb{1}$.
The holomorphic implicit function theorem (see e.g. \cite{FS}, Appendix B.5 and refs. therein) implies that the fixpoint equation 
(\ref{eq:zpzm}) has a unique holomorphic solution $Z_\pm(g)$ in a neighborhood of
$g = 0$. Let $g_0$ be the radius of convergence of the Taylor series of
$Z_+(g)$. Since the Taylor coefficients of $Z_+$ are non-negative, $g=g_0$
is a singularity of $Z_+(g)$ by Pringsheim's Theorem (\cite{FS} Thm.IV.6). Setting
\begin{equation}
	Z_+(g_0) = \lim_{g \nearrow g_0} Z_+(g)
\end{equation}
we have that $Z_+(g_0) < + \infty$. In fact, if $\xi = \infty$ this follows
from (\ref{eq:zpzm}), since $\varphi(Z_+)$ increases faster than linearly at $+ \infty$,
assuming that $p_n > 0$ for some $n \geq 2$. If $\xi < + \infty$ we must have 
$Z_\pm(g_0) \leq \xi$, because otherwise there would exist $0 < g_1 < g_0$ such
that $Z_+(g_1) = \xi$ and $Z_-(g_1) \leq \xi$ (or vice versa), contradicting (\ref{eq:zpzm})
(the LHS would be analytic at $g_1$ and the RHS not). In particular, we also have 
$g_0 < + \infty$ and that $g_0$ equals the radius of convergence for the Taylor 
series of $Z_-(g)$ by (\ref{eq:zpzm}).

The \textit{genericity assumption} mentioned above states that
\begin{equation}
	Z_\pm(g_0) < \xi\,,
\label{eq:genass}
\end{equation}
where we recall that $\xi$ is the radius of convergence of
(\ref{eq:generat_weight}). As mentioned previously, (\ref{eq:genass})
trivially holds if $\varphi(z)$ is a polynomial. In particular, for the random binary tree, where
\begin{equation}
	\varphi(z) = 1 + z^2,
\end{equation}
one finds that
\begin{equation}
	Z_\pm(\beta,0,g) = \frac{1 - \sqrt{1-8 g \cosh\beta}}{4g \cosh\beta}
\end{equation}
and hence $g_0 = \frac{1}{8 \cosh\beta}$ for $h=0$. On the other hand, for the uniform random tree with $p_n = 1$, $n \geq 0$, we have
\begin{equation}
	\varphi(z) = \frac{1}{1-z}
\end{equation}
and one finds
\begin{equation}
	Z_\pm(\beta,0,g) = \frac{1 - \sqrt{1-8 g \cosh\beta}}{2},
\end{equation}
and again $g_0 = \frac{1}{8 \cosh\beta}$ for $h=0$, which gives
\begin{equation}
	Z_\pm(g_0) = \frac{1}{2} < \xi = 1.
\end{equation}

\begin{rem}\label{rem:zero_field}
It should be noted that, in the absence of an external magnetic field,
i.e. for $h = 0$, one has $Z_+(\beta,0,g)=Z_-(\beta,0,g) \equiv
\bar{Z}(\beta,g)$ and the system (\ref{eq:zpzm}) determining $Z_\pm$
reduces to the single equation $\bar{Z} = 2g \cosh\beta \,
\varphi(\bar{Z})$. On the other hand, the same equation characterizes the
random tree models considered in \cite{DJW2} except for a rescaling of the
coupling constant $g$ by the factor $2\cosh\beta$. It follows that the
condition \eqref{eq:genass} can be considered as a generalization of
the genericity condition introduced in \cite{DJW2}. For this reason, the
results on the Hausdorff dimension and the spectral dimension
established in this paper follow from \cite{DJW2} in the case  $h=0$.

That the Ising model, without external field and with free boundary
conditions, simply gives rise to a rescaling of the coupling constant
$g$ is also easily seen directly for the finite volume partition
functions $Z_N$. 
\end{rem}

We henceforth assume (\ref{eq:genass}) to hold. Then the implicit function theorem gives
\begin{equation}
	\det \left( \mathbb{1} - g_0 \, \Phi'_0 \right) = 0,
\label{eq:J-zero}
\end{equation}
where
\begin{equation}
	\Phi'_0 = \Phi'(Z_+^0 , Z_-^0) = 
		\begin{pmatrix}
			a \, \varphi'(Z_+^0) & a^{-1} \, \varphi'(Z_-^0)\\
			b \, \varphi'(Z_+^0) & b^{-1} \, \varphi'(Z_-^0)
\end{pmatrix},
\label{eq:phiprime-def}
\end{equation}
with $Z_\pm^0 = Z_\pm(g_0)$. This justifies the first statement of the following proposition.
\begin{prop}
Assuming (\ref{eq:genass}), the system of equations
\begin{align}
	Z_+ &= g (a \, \varphi(Z_+) + a^{-1} \, \varphi(Z_-)) \label{eq:Zp_func_eq} \\
	Z_- &= g (b \, \varphi(Z_+) + b^{-1} \, \varphi(Z_-)) \label{eq:Zm_func_eq}\\
	0 &=  \det \left( \mathbb{1} - g \, \Phi'(Z_+,Z_-) \right) \label{eq:g0_func_eq}
\end{align} 
admits a solution $(Z_+^0,Z_-^0,g_0)$, with $Z_\pm^0 =
Z_\pm(g_0)$. Moreover, there exists $r >0$ such that the functions
$Z_\pm(g)$ have a representation of the form 
\begin{equation}
	Z_\pm(g) = Z_\pm^0 - \sqrt{K_\pm (g_0 - g)}
\label{eq:Zpm_squareroot}
\end{equation} 
in $\set{g \st |g-g_0| < r \text{ and } \arg(g-g_0) \neq 0}$. Here the constants $K_\pm$ (depending only on $\beta$ and $h$) are given by
\begin{equation}
	K_+ = \alpha^2 K_-
\label{eq:Kpm}
\end{equation}
with
\begin{equation}
	\alpha \equiv \frac{g_0 \, a^{-1} \, \varphi'(Z_-^0)}{1 - g_0 \, a \, \varphi'(Z_+^0)} = \frac{1 - g_0 \, b^{-1} \, \varphi'(Z_-^0)}{g_0 \, b \, \varphi'(Z_+^0)},
\label{eq:alpha}
\end{equation}
and
\begin{equation}\label{eq:Km}
	K_- \equiv \frac{2}{g_0} \, \frac{\alpha \, a \, \varphi(Z_+^0) + b^{-1} \, \varphi(Z_-^0)}{\alpha^3 \, a \, \varphi''(Z_+^0) + b^{-1} \, \varphi''(Z_-^0)}\,. 
\end{equation}
\end{prop}

\begin{proof}
The square-root behaviour (\ref{eq:Zpm_squareroot}) is a direct consequence of Thm. 2.33 in \cite{Drmota}, provided that the the matrix $g_0 \Phi'_0$ is irreducible. This follows from the fact that $a,b >0$ and $\varphi'(Z_\pm^0)>0$.

The constants $K_\pm$ are easily obtained by expanding eqs. (\ref{eq:Zp_func_eq})-(\ref{eq:Zm_func_eq}) around $Z_\pm^0$ and using (\ref{eq:J-zero}). The identity in eq. (\ref{eq:alpha}) follows from (\ref{eq:J-zero}).
\end{proof}

\begin{rem} \label{rem:c-alpha}
The transpose of the matrix $g_0 \Phi'_0$ has positive entries and eigenvalue 1 with left eigenvector $c=(c_1 \; c_2)$, where
\begin{equation}
	c_1 = g_0 \, b \, \varphi'(Z_+^0), \quad c_2 = 1 - g_0 \, a \, \varphi'(Z_+^0).
\label{eq:c-def}
\end{equation}
The other eigenvalue is given by
\begin{equation}
	\lambda = \det g_0 \Phi'(Z_+^0,Z_-^0) = 2g_0^2\sinh (2\beta)\varphi'(Z_+^0) \varphi'(Z_-^0).
\label{eq:lambda}
\end{equation}
In particular, we have $\lambda < 1$ by construction and $\lambda > -1$ since
\begin{equation}
	1 + \lambda = g_0 (a \varphi'(Z_+^0) + b^{-1} \varphi'(Z_-^0)) > 0.
\end{equation}
Hence 1 is the Perron-Frobenius eigenvalue of the transpose of $g_0
\Phi'_0$ (cf. \cite{FS} and refs. therein) and we have $c_1,c_2 > 0$
and accordingly $\alpha >0$. 
\end{rem}

The above result allows us to use a standard transfer theorem \cite{FS} to
determine the asymptotic behavior of $Z_{N\pm}(\beta,h)$ for $N \rightarrow \infty$.
In order to simplify the statement of the results we consider only the aperiodic case, where
$\set{n \st p_n >0}$ has greatest common divisor 1.

\begin{prop}
\label{cor:Zn}
Suppose the greatest common divisor of $\set{ n \st p_n >0 }$ is 1. Then, we have
\begin{equation}
	Z_{N\pm}(\beta,h) = \frac{1}{2} \, \sqrt{\frac{g_0 K_\pm}{\pi}} \, g_0^{-N} \, N^{-3/2} (1+o(1))
\label{eq:znpm}
\end{equation}
for $N \rightarrow \infty$, where $g_0$, $K_\pm >0$ are determined by
(\ref{eq:J-zero}-\ref{eq:phiprime-def}) and (\ref{eq:Kpm}-\ref{eq:Km}). 
\label{cor:znpm}
\end{prop}

\begin{proof}

We first show that $g=g_0$ is the only singularity of $Z_\pm(g)$ along
the circle of radius $g_0$ centered at 0, hence the functions
$Z_\pm(g)$ can be analytically extended to a disc of radius $g_0 +
\epsilon$, except for a slit from $g_0$ to $g_0 + \epsilon$, for some
$\epsilon>0$. 

From  $\det(\mathbb{1} - g \Phi'(Z_+,Z_-))|_{g=g_0} = 0$ and
$\det(\mathbb{1} - g \Phi'(Z_+,Z_-))|_{g=0} = 1$, we have 
\begin{equation}
		\det(\mathbb{1} - g \Phi'(Z_+,Z_-)) > 0, \qquad  \quad  0 \leq g < g_0.
\end{equation}
Hence
\begin{equation}
	|\det(\mathbb{1} - g \Phi'(Z_+,Z_-))| \geq
        \det\left(\mathbb{1} - |g| \Phi'(Z_+(|g|),Z_-(|g|))\right) > 0
\end{equation}
for $|g|<g_0$, where we have used that $\varphi$ and $Z_\pm$ have 
positive Taylor coefficients. Moreover, in the limiting case $|g|=g_0$ 
we get that $\det(\mathbb{1} - g \Phi'(Z_+,Z_-))=0$ if and only if
\begin{equation}
	g \varphi'(Z_\pm(g)) = g_0 \varphi'(Z_\pm(g_0)).
\end{equation}
In particular, $|\varphi'(Z_\pm(g))|=\varphi'(Z_\pm(g_0))$ which implies
\begin{equation}
	|Z_\pm(g)| = Z_\pm(g_0).
\label{eq:zpmg}
\end{equation}
By the definition of $Z_{N\pm}(\beta,h)$ we have that $Z_{N\pm}(\beta,h) > 0$ for all $N$ of the form
\begin{equation}
	N = 1 + n_1 + n_2 + \cdots + n_s,
\end{equation}
where $n_i$ are such that $p_{n_i} > 0$, $i = 1, \ldots, s$. Hence, eq. (\ref{eq:zpmg}) implies
\begin{equation}
	g^N = e^{i \theta} g_0^N
\end{equation}
for some fixed $\theta \in \mathbb{R}$ and all such $N$. By the assumption on $(p_n)$ this implies $g=g_0$. This proves our claim.

The result is then a direct consequence of Thm. VI.4 in \cite{FS}.
\end{proof}

\subsection{The limiting measure}

For $1 \leq N < \infty$ and fixed $\beta$, $h \in \mathbb{R}$ we define
the probability distributions $\mu_{N\pm}$ on $\Lambda $, supported on $\Lambda_{N\pm}$, by
\begin{equation}
	\mu_{N\pm}(\tau_s) = \frac{1}{Z_{N\pm}} e^{-H(\tau_s)} \rho(\tau), \qquad \tau_s \in \Lambda_{N\pm},
\end{equation} 
such that
\begin{equation}
	\mu_N = \frac{Z_{N+}}{Z_N} \mu_{N+} + \frac{Z_{N-}}{Z_N} \mu_{N-}.
\label{eq:mn-mnp-mnm}
\end{equation}

We shall need the following proposition, that can be obtained by a
slight modification of the proof of Proposition 3.2 in \cite{BD03},
and whose details we omit.

\begin{prop}
Let $K_R$, $R \in \mathbb{N}$, be a sequence of positive numbers. Then the subset
\begin{equation}
	C = \bigcap_{R=1}^\infty \set{ \tau_s \in \Lambda \st |B_R(\tau)| \leq K_R }
\end{equation}
of $\Lambda$ is compact.
\label{prop:compact}
\end{prop}

We are now ready to prove the following main result of this section.

\begin{thm}
\label{thm:wl}
Let $\beta, h \in \mathbb{R}$ and assume that the genericity condition (\ref{eq:genass}) holds and that the greatest common divisor of $\set{ n \st p_n >0 }$ is 1. Then the weak limits
\begin{equation}
	\mu_{\pm} = \lim_{N \rightarrow \infty} \mu_{N\pm} \quad \text{and} \quad \mu = \lim_{N \rightarrow \infty} \mu_N
\label{eq:wl-pm}
\end{equation}
exist as probability measures on $\Lambda$ and
\begin{equation}
	\mu = \frac{\alpha}{1 + \alpha} \, \mu_+ + \frac{1}{1 + \alpha} \, \mu_-,
\label{eq:wl}
\end{equation}
where $\alpha$ is given by (\ref{eq:alpha}).
\end{thm}
\begin{proof}
The identity (\ref{eq:wl}) follows immediately from (\ref{eq:mn-mnp-mnm}), Corollary \ref{cor:znpm} and (\ref{eq:Kpm}),
provided $\mu_\pm$ exist. Hence it suffices to show that $\mu_+$ exists
(since existence of $\mu_-$ follows by identical arguments).

According to \cite{BD03}, it is sufficient to prove that the sequence
$(\mu_{N+})$ satisfies a certain  \textit{tightness condition} (see
e.g. \cite{PB68} for a definition) and that the sequence 
\begin{equation}
	\mu_{N+}(\set{ \tau_s \st B_R(\tau) = \hat{\tau}, s|_{V(\hat{\tau})} = \hat{s}})
\end{equation}
is convergent in $\mathbb{R}$ as $N \to \infty$, for each finite tree $\hat{\tau} \in \mathcal{T}$ and fixed spin configuration $\hat{s}$.
\smallskip

\noindent\textit{Tightness of $(\mu_{N+})$}: As a consequence of
Proposition \ref{prop:compact}, this condition holds if we
show that for each $\epsilon > 0$ and $R \in \mathbb{N}$ there exists $K_R >0$
such that
\begin{equation}
	\mu_{N+}(\set{ \tau_s \st |B_R(\tau)| > K_R }) < \epsilon, \quad N \in \mathbb{N}.
\label{eq:tight}
\end{equation}
For $R=1$ this is trivial. For $R=2$, $k \geq 1$ we have
\begin{equation}
	\begin{split}
		\mu_{N+} & (\set{ \tau_s \st |B_2(\tau)| = k+1 })\\
		 & = Z_{N+}^{-1} \sum_{N_1 + \cdots + N_k = N-1} \left[ a \prod_{i=1}^kZ_{N_i+} + a^{-1}\prod_{i=1}^k Z_{N_i-}\right]p_k\\
		 & \leq k \, \sum_{\substack{N_1 + \cdots + N_k = N-1 \\ N_1 \geq (N-1)/k}} Z_{N+}^{-1} \left[ a \prod_{i=1}^{k} Z_{N_i+} + a^{-1} \prod_{i=1}^{k} Z_{N_i-} \right] p_k\\
		 & \leq cst. \; k^{5/2} \left[ Z_+(g_0)^{k-1} + Z_-(g_0)^{k-1} \right] p_k\,,
	\end{split}
\end{equation}
where we have used (\ref{eq:znpm}). The last expression tends to zero for $k \to \infty$ as
a consequence of (\ref{eq:genass}). This proves (\ref{eq:tight}) for $R=2$.

For $R>2$ it is sufficient to show
\begin{equation}
	\mu_{N+}(\set{ \tau_s \st |B_{R+1}(\tau)| > K, B_R(\tau)=\hat{\tau}, s|_{V(\hat{\tau})} = \hat{s} }) \rightarrow 0
\label{eq:tightr}
\end{equation}
uniformly in $N$ for $k \rightarrow \infty$, for fixed $\hat{\tau}$ of height $R$ and fixed
$\hat{s} \in \set{ \pm 1 }^{V(\hat{\tau})}$, as well as fixed $K>0$. Let $L$ denote the number of vertices in $\hat{\tau}$
at maximal height $R$. Any $\tau \in \Lambda$ with $B_R(\tau) = \hat{\tau}$ is
obtained by attaching a sequence of trees $\tau_1, \ldots, \tau_S$ in
$\Lambda$ such that the root vertex of $\tau_i$ is identified with a vertex at maximal
height in $\hat{\tau}$. We must then have
\begin{equation}
	|\tau_1| + \cdots + |\tau_S| = |\tau| - |\hat{\tau}|
\end{equation}
and
\begin{equation}
	k_1 + \cdots + k_L = S,
\end{equation}
where $k_i \geq 0$ denotes the number of trees attached
to vertex $v_i$ in $\hat{\tau}$, $i=1,\ldots,L$. For fixed $k_1, \ldots, k_L$ we get a contribution to
(\ref{eq:tightr}) equal to
\begin{equation}
	\begin{split}
		Z_{N+}^{-1} & \sum_{N_1+\cdots+N_S = N-|\hat{\tau}|} \left( \prod_{i=1}^{L} \prod_{j=1}^{k_i}(Z_{N_{k_1 + \cdots + k_{i-1}+j} \, \hat{s}_{v_i}})^{k_i} p_{k_i}\right) e^{-H(\hat{\tau}_{\hat{s}})} \prod_{v \in V(\hat{\tau})\setminus \set{r,v_1,\ldots,v_L}} p_{\sigma_v-1}\\
		& \leq \text{const} \prod_{i=1}^L (\max Z_\pm^0)^{k_i} p_{k_i} \, (k_i+1)^{5/2}		
	\end{split}
\label{eq:tightr-ineq}
\end{equation} 
where the inequality is obtained as above for $R=2$ and the constant is independent of $k_1,\ldots,k_L$.

Since
\begin{equation}
	|B_{R+1}(\tau)| = |\hat{\tau}| + k_1 + \cdots + k_L >K
\end{equation}
and the number of choices of $k_1, \ldots, k_L \geq 0$ for fixed
$k = k_1 + \cdots + k_L$ equals
\begin{equation}
	\binom{k + L-1}{L-1} \leq \frac{k^{L-1}}{(L-1)!}
\end{equation}
the claim (\ref{eq:tightr}) follows from (\ref{eq:genass}) and (\ref{eq:tightr-ineq}).
\smallskip

\noindent \textit{Convergence of} $\mu_{N+}(\set{ \tau_s \st B_R(\tau) = \hat{\tau}, s|_{V(\hat{\tau})} = \hat{s} })$:
Using the decomposition of $\tau$ into $\hat{\tau}$ with
branches described above and using the arguments in
the last part of the proof of Theorem 3.3 in \cite{BD03} we get, with notation as above, that
\begin{equation}\label{eq:vol}
	\begin{split}
		\mu_{N\pm}&(\set{ \tau_s \st B_R(\tau) = \hat{\tau}, s|_{V(\hat{\tau})} = \hat{s} })\\
		&\xrightarrow{N\rightarrow \infty} \frac{g_0^{|\hat{\tau}|}}{\sqrt{K_\pm}} \, e^{-H(\hat{\tau}_{\hat{s}})} \sum_{i=1}^L \sqrt{K_{\hat{s}(v_i)}} \varphi'(Z_{\hat{s}(v_i)}^0) \prod_{j \neq i} \varphi(Z_{\hat{s}(v_j)}^0),
	\end{split}
\end{equation}
provided $\hat{s}(r)= \pm 1$ (if $\hat{s}(r)= \mp 1$ the limit is trivially 0).
\end{proof}

\bigskip

Introducing the notation
\begin{equation}
	A(\hat{s}) = \set{ \tau_s \st B_R(\tau) = \hat{\tau}, \, s|_{V(\hat{\tau})} = \hat{s}}\,,
\end{equation}
where $\hat{\tau}$ is a finite tree of height $R$ with spin
configuration $\hat{s}$, and using (\ref{eq:Kpm}), it follows from \eqref{eq:vol} that the $\mu_\pm$-volumes of this set are given by 
\begin{equation}
	\mu_\pm (A(\hat{s})) = g_0^{|\hat{\tau}|} \, e^{-H(\hat{\tau}_{\hat{s}})} \sum_{i=1}^L \alpha^{(\hat{s}(v_i) \mp 1)/2} \, \varphi'(Z_{\hat{s}(v_i)}^0) \prod_{j \neq i} \varphi(Z_{\hat{s}(v_j)}^0),
\end{equation}
if $\hat{s}(r) = \pm 1$ and where $v_1, \ldots, v_L$ are the vertices
at maximal distance from the root $r$ in $\hat{\tau}$.

The above calculations show, by similar arguments as in \cite{BD03,CD}, that the limiting measures $\mu_\pm$ are concentrated
on trees with a single infinite path starting at $r$, called the
\emph{spine}, and attached to each spine vertex $u_i$, $i=1,2,3\dots$,
is a finite number $k_i$ of finite trees, called \emph{branches}, some
of which are attached to the left and some to the right as seen from the root, cf. Fig.\ref{fig:grt}.

The following corollary provides a complete description of the limiting measures $\mu_\pm$.

\begin{cor}\label{cor:measures-description}
The measures $\mu_\pm$ are concentrated on the sets
\begin{equation}
	\bar{\Lambda}_{\pm} = \set{ \tau_s \in \Lambda_\pm \st \tau \, \text{has a single spine} },
\end{equation}
respectively, and can be described as follows:
\begin{enumerate}[i)]
	\item The probability that the spine vertices $u_0 = r$, $u_1,u_2,\ldots,u_N$
	have \, $k'_1,\ldots,k'_N$ left branches and $k''_1,\ldots,k''_N$ right branches
	and spin values $s_0 = \pm 1$, $s_1, s_2, \ldots, s_N$, respectively, equals
		\begin{equation}
			\begin{split}
				 \rho^{s_0}_{k'_1,\ldots,k'_N,k''_1,\ldots,k''_N} & (s_0,\ldots,s_N)\\
				& = g_0^N e^{-H_N} \left( \prod_{i=1}^N (Z_{s_i}^0)^{k'_i+k''_i} p_{k'_i+k''_i+1} \right) \alpha^{(s_N - s_0)/2},
			\end{split}
		\label{eq:prob-spine}
		\end{equation}
	with
	\begin{equation}
		H_N = - \beta \sum_{i=1}^N s_{i-i} s_i - h \sum_{i=1}^N s_i.
	\end{equation}
	\item The conditional probability distribution of any
	finite branch $\tau_s$ at a fixed $u_i$, $1 \leq i \leq N$, given $k'_1,\ldots,k'_N$, 
	$k''_1,\ldots,k''_N$, $s_0,\ldots,s_N$ as above, is given by
	\begin{equation}
		\nu_{s_i}(\tau_s) = (Z_{s_i}^0)^{-1} \, g_0^{|\tau|} \, e^{-H(\tau_s)} \prod_{v \in V(\tau) \setminus u_i} p_{\sigma_v-1}
	\label{eq:prob-branch}
	\end{equation}
	for $s(u_i)=s_i$, and 0 otherwise.
	\item The conditional distribution of the infinite branch at $u_N$, given
	$k'_1, \ldots,$ $k'_N$, $k''_1,\ldots,k''_N$, $s_0,\ldots,s_N$, equals $\mu_{s_N}$.
\end{enumerate}
\end{cor}

\section{Hausdorff and spectral dimensions}\label{sec:dim}

In this section we determine the values of the Hausdorff and spectral
dimensions of the ensemble of trees $(\mathcal{T},\bar{\mu})$ obtained from
$(\Lambda,\mu)$ by integrating over the spin degrees of freedom, that is
\begin{equation}
	\bar{\mu}(A) = \mu( \set{ \tau_s \st \tau \in A } )
\end{equation}
for $A \subseteq \mathcal{T}$. Note that the mapping 
$\tau_s \rightarrow \tau$ from $\Lambda$ to $\mathcal{T}$ is a 
contraction w.\,r.\,t. the metrics (\ref{eq:metric-Lambda}) and
(\ref{eq:metric-T}).

Most of the arguments in this section are based on the methods of
\cite{DJW2}, and we shall mainly focus on the novel ingredients
that are needed and otherwise refer to \cite{DJW2} for additional details. 

\subsection{The annealed Hausdorff dimension}
\label{sec:haus-d}

\begin{thm} \label{thm:Haus-dim}
Under the assumptions of Theorem \ref{thm:wl} the annealed Hausdorff dimension 
of $\bar{\mu}$ is 2 for all $\beta$, $h$:
\begin{equation}
	\bar{d}_h = \lim_{R \rightarrow \infty} \frac{\ln  \left< |B_R| \right>_{\bar{\mu}}}{\ln R} = 2\,.
\end{equation}
\end{thm}
\begin{proof}
Consider the probability distribution $\nu_\pm$ on 
$\set{ \tau_s \st \tau \, \text{is finite} }$ given by (\ref{eq:prob-branch}) and
denote by $D_R(\tau)$ the set of vertices at distance $R$ from the root 
in $\tau$. For a fixed branch $T$, we set 
\begin{equation}
	f_R^{\pm} = \left< |D_R| \right>_{\nu_\pm} \, Z_\pm^0.
\end{equation}
where $\langle \cdot \rangle_{\nu_\pm}$ denotes the  expectation value w.r.t. $\nu_\pm$. Arguing as in the derivation of (\ref{eq:zpzm}), we find 
\begin{equation}
	\begin{cases}
		f_R^+ = g_0 \, \left(a \, \varphi'(Z_+^0) \, f_{R-1}^+ + a^{-1} \, \varphi'(Z_-^0) \, f_{R-1}^-\right)\\
		f_R^- = g_0 \, \left(b \, \varphi'(Z_+^0) \, f_{R-1}^+ + b^{-1} \, \varphi'(Z_-^0) \, f_{R-1}^-\right) \,,
	\end{cases}
\label{eq:DR_branch_sys}
\end{equation}
for $R\geq 2$, and $f_1^{\pm} = Z_\pm^0$.
Using that $c$, given by (\ref{eq:c-def}), is a left eigenvector of $g_0 \Phi'_0$ with eigenvalue 1, this implies
\begin{equation}
	\begin{split}
		c_1 \, f_R^+ + c_2 \, f_R^- & = c_1 \, f_{R-1}^+ + c_2 \, f_{R-1}^- = \ldots \\ 
		& = c_1 \, f_1^+ + c_2 \, f_1^- = c_1 \, Z_+^0 + c_2 \, Z_-^0\,.
	\end{split}	
\end{equation}
Since $c_1$, $c_2$, $Z_\pm^0$, $f_R^\pm > 0$, we conclude that
\begin{equation}
	k_1 \leq \left< |D_R| \right>_{\nu_\pm} \leq k_2, \quad R \geq 1\,,
\label{eq:fR-bound}
\end{equation}
where $k_1$, $k_2$ are positive constants (depending on $\beta$, $h$). Using
\begin{equation}
	\left< |B_R| \right>_{\nu_\pm} = \sum_{R'=0}^R \left< |D_{R'}| \right>_{\nu_\pm}
\end{equation}
we then obtain
\begin{equation}
\label{eq:ballR-bound}
	1 + k_1 \, R \leq \left< |B_R| \right>_{\nu_\pm} \leq 1+ k_2 \, R,
\end{equation}
Finally, it follows from (\ref{eq:prob-spine}) that
\begin{equation}\label{volball}
	1+R + k_1 \, \frac{1}{2} R(R+1) \leq \left< |B_R| \right>_\mu \leq 1 + R + k_2 \, \frac{1}{2} R(R+1)\, ,
\end{equation}
which proves the claim.
\end{proof}

\begin{rem}
By a more elaborate argument, using the methods of \cite{DJW2,DJW3}, one can
show that the Hausdorff dimension $d_h$ defined by \eqref{Haus1}
exists and equals $2$ almost surely, that is for all trees $\tau\in
\cal{T}$ except for a set of vanishing $\bar\mu$-measure. We shall not
make use of this result below and refrain from giving further details
in this paper.
\end{rem}

\subsection{The annealed spectral dimension}
\label{sec:spect-d}

In this section we first establish two results needed for determining the spectral dimension. The first one is a version of a classical result, proven by Kolmogorov for Galton-Watson trees \cite{Harris}, on survival probabilities for $\nu_\pm$.

\begin{prop}
\label{prop:Kolm-gen}
The measures $\nu_\pm$ defined by (\ref{eq:prob-branch}) fulfill
\begin{equation}
\frac{k_-}{R}\leq	\nu_\pm(\set{\tau_s \in \Lambda\st D_R(\tau) \neq \emptyset}) \leq \frac{k_+}{R}, \quad R \geq 1,
\end{equation}
where $k_\pm>0$ are constants depending on $\beta, h$.
\end{prop}
\begin{proof}
Let $H_R^\pm(w)$ be the generating function for the distribution of $|D_R|$ w.r.t. $\nu_\pm$,
\begin{equation}
	H_R^\pm(w) = Z_\pm^0 \sum_{n=0}^\infty \nu_\pm(\set{\tau_s \st |D_R(\tau)| = n}) \, w^n.
\end{equation}
Arguing as in the proof of (\ref{eq:zpzm}), we have
\begin{equation}
\label{eq:HR-system}
\begin{cases}
\begin{aligned}
	H_R^+ & = g_0 \left( a \, \varphi(H_{R-1}^+) + a^{-1} \, \varphi(H_{R-1}^-) \right)\\
	H_R^- & = g_0 \left( b \, \varphi(H_{R-1}^+) + b^{-1} \, \varphi(H_{R-1}^-) \right)\,,
\end{aligned}
\end{cases}
\end{equation}
for $R\geq 2$, and $H_1^{\pm} = Z_\pm^0\,w$.

Note that
\begin{equation}
	Z_\pm^0 \, \nu_\pm(\set{\tau_s \in \Lambda \st D_R(\tau) \neq \emptyset}) = Z_\pm^0 - H_R^\pm(0),
\label{eq:Z-HR}
\end{equation}
and that the radius of convergence for $H_R^\pm$ is $\geq 1$.
Also, $(H_R^\pm(0))_{R \geq 1}$ is an increasing sequence. In fact, $H_1^\pm(0)=0$ and so $H_2^\pm(0) > 0$
by (\ref{eq:HR-system}). Since $\varphi$ is positive and increasing on $[0,\xi)$, it then follows by induction from (\ref{eq:HR-system}) that $(H_R^\pm(0))_{R \geq 1}$ is increasing. Hence, we conclude from (\ref{eq:HR-system}) and (\ref{eq:zpzm}) that
\begin{equation}
	H_R^\pm(0) \nearrow Z_\pm^0 \quad \text{for} \quad R \rightarrow \infty.
\end{equation}
Taking $R$ large enough and expanding $\varphi(H_R^\pm(0))$ around $Z_\pm^0$ we obtain, in matrix form,
\begin{equation}
	\Delta_R = g_0 \, \Phi'_0 \, \Delta_{R-1} - \frac{g_0}{2} \, \Phi''_0 \, \Delta_{R-1}^2 + O(\Delta_{R-1}^3)\,,
\label{eq:HR-exp}
\end{equation}
where
\begin{equation}
	\Delta_R^n = 
	\begin{pmatrix}
	(\Delta_R^+)^n\\
	(\Delta_R^-)^n
	\end{pmatrix}
	=
	\begin{pmatrix}
	(Z^0_+ - H^+_R(0))^n\\
	(Z^0_+ - H^+_R(0))^n
	\end{pmatrix},
\end{equation}
$\Phi'_0$ is given by (\ref{eq:phiprime-def}) and
\begin{equation}
	\Phi''_0 = 
		\begin{pmatrix}
			a \, \varphi''(Z_+^0) & a^{-1} \, \varphi''(Z_-^0)\\
			b \, \varphi''(Z_+^0) & b^{-1} \, \varphi''(Z_-^0)
		\end{pmatrix}.	
\end{equation}

Setting $L_R = c\, \Delta_R$, eq. (\ref{eq:HR-exp}) gives
\begin{equation}
	L_R = L_{R-1} - \frac{g_0}{2} \, c \, \Phi''_0 \, \Delta_{R-1}^2 + O(\Delta_{R-1}^3)\,.
\end{equation}
From this we deduce that there exists $R_0>0$ such that
\begin{equation}
	L_{R-1} - A_- L_{R-1}^2 \leq L_R \leq L_{R-1} - A_+ L_{R-1}^2, \qquad R \geq R_0,
\end{equation}
where $A_\pm = A_\pm(\beta,h)$ are constants. Hence, it follows that
\begin{equation}
	\frac{1}{L_{R-1}} + B_- \leq \frac{1}{L_{R-1}} \frac{1}{1-A_- L_{R-1}} \leq \frac{1}{L_R} \leq \frac{1}{L_{R-1}} \frac{1}{1-A_+ L_{R-1}} \leq \frac{1}{L_{R-1}} + B_+,
\end{equation}
for $R \geq R_0$, where $B_\pm >0$ are constants. This implies
\begin{equation}
	B_- R + C_- \leq \frac{1}{L_R} \leq B_+ R + C_+
\end{equation}
for suitable constants $C_\pm$. Evidently, this proves that
\begin{equation}
	\frac{D_-}{R} \leq Z_\pm^0 - H_R^\pm(0) \leq \frac{D_+}{R}, \quad R \geq 1,
\end{equation}
where $D_\pm >0$ are constants, which together with (\ref{eq:Z-HR}) proves the claim.
\end{proof}

We also note the following generalization of Lemma 4 in \cite{DJW2}.

\begin{lem}
\label{lem:int-full-branch}
Suppose $u \, : \, \Lambda \rightarrow \mathbb{C}$ is a bounded function depending only on $\tau_s \in \Lambda$ through the ball $B_R(\tau)$ and the spins in $B_R(\tau)$, except those on its boundary, for some $R \geq 1$. Moreover, define the function $E_R \, : \, \Lambda \rightarrow \mathbb{R}$ by
\begin{equation}
	E_R(\tau_s) = \sum_{v \in D_R(\tau)} \frac{\sqrt{K_{s_v}}}{Z_{s_v}^0},
\end{equation}
with the convention $E_R(\tau_s) = 0$ if $D_R(\tau) = \emptyset$. Then
\begin{equation}
	\int_\Lambda u(\tau_s) d\mu_\pm(\tau_s)=  \frac{Z_{\pm}^0}{\sqrt{K_\pm}} \int_\Lambda u(\tau_s) E_R(\tau_s) d\nu_\pm(\tau_s).
\label{eq:int-full-branch}
\end{equation}
\end{lem}
\begin{proof}
Using (\ref{eq:prob-spine}-\ref{eq:prob-branch}) we may evaluate the LHS of (\ref{eq:int-full-branch}) and get
\begin{equation}
	\sum_{\tau_s \in \Lambda(R)} u(\tau_s) \, g_0^{|\tau|} \,
        e^{-H (\tau_s)}  \, \alpha^{(s(v_R)-s_0)/2}\prod_{v \in
          V(\tau)\setminus r} p_{\sigma_v-1}\,, 
\end{equation}
where $\Lambda(R)$ denotes the set of finite rooted trees in $\Lambda$
with one marked vertex $w_R$ of degree 1 at distance $R$ from the root,
and $v_R$ is the neighbor of $w_R$.

On the other hand, the integral on the RHS can be written as
\begin{equation}
	\frac{1}{Z_\pm^0} \sum_{\tau_s \in \Lambda(R)} u(\tau_s) \,
        g_0^{|\tau|} e^{-H(\tau_s)}
        \frac{\sqrt{K_{s(v_R)}}}{Z_{s(v_R)}^0} Z_{s(v_R)}^0 \prod_{v
          \in V(\tau)\setminus r} p_{\sigma_v-1}. 
\end{equation}
By comparing the two expressions the identity
(\ref{eq:int-full-branch}) follows.
\end{proof}

As a consequence of this result we have the following lemma.

\begin{lem}
\label{lem:BRinverse-bound}
There exist constants $c_\pm > 0$ such that
\begin{equation}
	\left< |B_R|^{-1} \right>_{\mu_\pm} \leq c_\pm R^{-2}
\end{equation}
\end{lem}
\begin{proof}
Define, for fixed $R \geq 1$, the function
\begin{equation}
	u(\tau) = 
	\begin{cases}
	|D_R(\tau)|^{-1} \quad & \text{if} \; D_R(\tau) \neq \emptyset\\
	0 & \text{otherwise}.
	\end{cases}
\end{equation}
Then $u(\tau)$ fulfills the assumptions of Lemma \ref{lem:int-full-branch} for this 
value of $R$. Hence
\begin{equation}
\begin{split}
	\left< |D_R(\tau)|^{-1} \right>_{\mu_\pm} & = \frac{Z_\pm^0}{\sqrt{K_\pm}} \sum_{\tau_s:D_R(\tau) \neq \emptyset}  |D_R(\tau)|^{-1} \, E(\tau_s) \, e^{-H(\tau_s)} \prod_{v \in V(\tau) \setminus r} p_{\sigma_v-1}\\
	& \leq c'_\pm \sum_{\tau_s:D_R(\tau) \neq \emptyset} e^{-H(\tau_s)} \prod_{v \in V(\tau) \setminus r} p_{\sigma_v-1} \leq \frac{c''_\pm}{R},
\end{split}
\label{eq:DR_inv_bound}
\end{equation}
where Proposition \ref{prop:Kolm-gen} is used in the last step. Combining
this fact with Jensen's inequality, we obtain
\begin{equation}
\begin{split}
	\left< |B_R|^{-1} \right>_{\mu_\pm} & = \left< \frac{1}{|D_1| + \cdots + |D_R|} \right>_{\mu_\pm} \\
	& \leq R^{-1} \left< (|D_1| |D_2| \cdots |D_R|)^{-1/R}\right>_{\mu_\pm} \\
	& \leq R^{-1} \prod_{i=1}^R \left< |D_i|^{-1} \right>_{\mu_\pm}^{1/R} \\
	& \leq c''_\pm R^{-1} (R!)^{-1/R} \leq c_\pm R^{-2}.
\end{split}
\end{equation}
\end{proof}

\bigskip

Returning to the spectral dimension, let us define, 
with the notation of subsection \ref{sec:spec-dim}, the
generating function for return probabilities of the simple random walk on a tree $\tau$ by
\begin{equation}
	Q_\tau(x) = \sum_{t=0}^\infty (1-x)^{\frac{t}{2}} \pi_t(\tau,r)\,,
\end{equation}
and set
\begin{equation}
	Q(x) = \left< Q_\tau(x) \right>_{\bar{\mu}}\,.
\end{equation}
The annealed spectral dimension as defined by (\ref{spectral2}) is
related to the singular behavior of the function $Q(x)$ as
follows. First, note that if $\bar d_s$ exists, we have
\begin{equation}
	\left< \pi_t(\tau,r) \right>_{\bar{\mu}} \sim t^{- \frac{\bar{d}_s}{2}}, \quad t \to \infty\,.
\label{eq:ds-larget}
\end{equation} 
For $\bar d_s<2$, this implies that $Q(x)$ diverges as
\begin{equation}
	Q(x) \sim x^{-\gamma}, \quad \text{as} \quad x \to 0,
\label{eq:Q-sim}
\end{equation}
where
\begin{equation}\label{dsa}
\gamma = 1-\frac{\bar d_s}{2}\,.
\end{equation}
We shall take \eqref{eq:Q-sim} and \eqref{dsa} as the definition of $\bar
d_s$ and prove \eqref{eq:Q-sim} with $\gamma =\frac 13$ by establishing
the estimates
\begin{equation}\label{Qest}
	\underline{c} \, x^{-1/3} \leq Q(x) \leq \bar{c} \, x^{-1/3}
\end{equation}
for $x$ sufficiently small, where $\underline{c}$ and
$\bar{c}$ are positive constants, that may depend on $\beta, h$.
\begin{thm}
\label{thm:spect-dim}
Under the assumptions of Theorem \ref{thm:wl}, the annealed spectral 
dimension of $(\mathcal{T},\bar{\mu})$ is
\begin{equation}
	\bar{d}_s = \frac{4}{3}\,.
\end{equation} 
\end{thm}
\begin{proof}
\noindent We first prove the lower bound in \eqref{Qest}. 

Let $R \geq 1$ be fixed and consider the spine vertices $u_0, u_1,
\dots, u_R$ with given spin values $s_0, \ldots, s_R$ and branching
numbers $k'_1,\ldots,k'_R$, $k''_1,\ldots,k''_R \geq 0$ as in
Corollary \ref{cor:measures-description}. The
conditional probability that a given branch at $u_j$ has length
$\geq R$ is bounded by $\frac{c}{R}$ by Proposition \ref{prop:Kolm-gen}. Hence, the conditional
probability that at least one of the $k'_j + k''_j$ branches at $u_j$ has height
$\geq R$ is bounded by $(k'_j + k''_j) \frac{c}{R}$. Using Corollary \ref{cor:measures-description} and 
summing over $k'_1,\ldots,k'_R$, $k''_1,\ldots,k''_R$, we get that the conditional
probability $q_R$ that at least one branch at $u_j$ is of height $\geq R$, given
$s_0,\ldots,s_R$, is bounded by
\begin{equation}
	\frac{1}{1 + \alpha} \, g_0^R \, e^{-H_R} \prod_{\substack{i=1\\i \neq j}}^R \varphi'(Z_{s_i}^0) \, \varphi''(Z_{s_j}^0) \alpha^{(s_R+1)/2} \, \frac{c}{R} \leq \frac{c'}{R}.
\end{equation}

Using that the distributions of the branches at different spine vertices are
independent for given $s_0,\ldots,s_R$, it follows that the conditional
probability that no branch at $u_1,\ldots,u_R$ has length $\geq R$, for given
$s_0,\ldots,s_R$, is bounded from below by
\begin{equation}
	(1-q_R)^R \geq \left( 1 - \frac{c'}{R} \right)^R \geq e^{-c'+O(R^{-1})}.
\end{equation}
Denoting this conditioned event by $\mathcal{A}_R$, it follows from Lemmas 6 and 7
in \cite{DJW2} that the conditional expectation of $Q_\tau(x)$,
given $s_0,s_1,\ldots,s_R$, is
\begin{equation}
\begin{split}
	\geq e^{c'+O(R^{-1})} \left< \left( \frac{1}{R} + R x + \sideset{}{^R}\sum_{T\subset\tau} x \, |T| \right)^{-1} \right>_R\\
	\geq e^{c'+O(R^{-1})}  \left( \frac{1}{R} + R x + x \left< \sideset{}{^R}\sum_{T\subset\tau} |T|\right>_R \right)^{-1}.
\end{split}
\end{equation}
Here $\left< \cdot \right>_R$ denotes the conditional expectation value
w.r.t. $\mu$ on $\mathcal{A}_R$ and
$\sideset{}{^R}\sum_{T\subset\tau}$ the sum over all branches $T$ of
$\tau$ attached to vertices on the spine at distance $\leq R$ from the
root. We have 
\begin{equation}
\begin{split}
	\left< \sideset{}{^R}\sum_{T\subset\tau} |T|\right>_R & = \sum_{i=1}^R \left< |B_R^i(\tau)| \right>_R \\
	&\leq \sum_{i=1}^R \mu(\mathcal{A}_R \st s_0,\ldots,s_R)^{-1} \left< |B_R^i| \right>_\mu \\
	&\leq e^{c'+O(R^{-1})} \sum_{i=1}^R \left< |B_R| \right>_{\nu_{s_i}}	\leq C \, R^2,
\end{split}
\end{equation}
where (\ref{eq:ballR-bound}) is used in the last step.

This bound being independent of $s_0,\ldots,s_R$ we have proven that
\begin{equation}
	Q(x) \geq cst. \, \left( \frac{1}{R} + R x + C R^2 x \right)^{-1}
\end{equation}
and consequently, choosing $R \sim x^{-\frac{1}{3}}$, it follows that
\begin{equation}
	Q(x) \geq \underline{c} \, x^{- \frac{1}{3}}.
\end{equation}

As concerns the upper bound in \eqref{Qest}, it follows by an argument
identical to the one in \cite{DJW2} on p.1245--50 by using Lemma
\ref{lem:BRinverse-bound}. 
\end{proof}

\section{Absence of spontaneous magnetization}
\label{sec:mag-prop}

Using the characterization of the measure $\mu^{(\beta,h)}$
established in Section \ref{sec:inf-measure} and that $\bar{d}_h = 2$,
we  are now in a position to discuss the magnetization properties of
generic Ising trees in some detail. In view of the fact that the trees have
a single spine, we distinguish between the magnetization on the spine
and the bulk magnetization. In subsection \ref{sec:spont-mag} we show
that the former can be expressed in terms of an effective Ising model
on the half-line $\set{0,1,2,\ldots}$. The bulk magnetization is
discussed in subsection \ref{sec:mean-mag} 

\subsection{Magnetization on the spine}
\label{sec:spont-mag}

The following result is crucial for the subsequent discussion.

\begin{prop} \label{prop:zpm-smooth}
Under the assumptions of Theorem \ref{thm:wl}, the functions $Z_\pm^0$  
are smooth functions of $\beta$, $h$.
\end{prop}
\begin{proof}
In Section \ref{sec:part-func} we have shown that $Z_\pm(\beta,h,g)$ is
a solution to the equation
\begin{equation}
	F(Z_+,Z_-,g) = 0
\end{equation}
where $F$ is defined in (\ref{eq:f}), and that
\begin{equation}
	Z_\pm^0(\beta,h) = Z_\pm(g_0(\beta,h),\beta,h)
\end{equation}
is a solution to
\begin{equation}
	\begin{cases}
		F(Z_+^0,Z_-^0,g_0) = 0\\
		\det(\mathbb{1} - g_0 \Phi'(Z_+^0,Z_-^0)) = 0\,,
	\end{cases}
\label{eq:zpzmg-zero}
\end{equation}
considered as three equations determining $(Z_+^0,Z_-^0,g_0)$ implicitly
as functions of $(\beta,h)$. Hence, defining $G: \; (-R,R)^2 \times \mathbb{R}^3 \rightarrow \mathbb{R}^3$ by
\begin{equation}
	G(Z_+^0,Z_-^0,g_0,\beta,h)=
		\begin{pmatrix}
	 		F(Z_+^0,Z_-^0,g_0) \\
    	\det(\mathbb{1} - g_0 \Phi'(Z_+^0,Z_-^0))   		
		\end{pmatrix}
\end{equation}
it suffices to show that its Jacobian $J$ w.r.t. $(Z_+^0,Z_-^0,g_0)$
is regular at $(Z_+^0(\beta,h),Z_-^0(\beta,h),g_0(\beta,h))$. We have
\begin{equation}
	J = 
	\begin{pmatrix}
		\mathbb{1} - g_0 \Phi'(Z_+^0,Z_-^0) & -\Phi(Z_+^0,Z_-^0)\\
		A_+ \qquad A_- & B		
	\end{pmatrix}\,,
\end{equation}
where
\begin{equation}
	A_\pm = \partd[]{Z_\pm^0}\det(\mathbb{1} - g_0 \Phi'(Z_+^0,Z_-^0))\,,  \qquad B = \partd[]{g_0}\det(\mathbb{1} - g_0 \Phi'(Z_+^0,Z_-^0)) 
\end{equation}
are readily calculated and equal
\begin{gather}
	A_+ =  - g_0 \, a \, \varphi''(Z_+^0) \, (1 - g_0 \, b^{-1} \varphi'(Z_-^0)) - g_0^2 \, a^{-1}\, b \, \varphi''(Z_+^0) \varphi'(Z_-^0)\,, \\
	A_- = - g_0 \, b^{-1} \, \varphi''(Z_-^0) \, (1 - g_0 \, a \varphi'(Z_+^0)) - g_0^2 \, a^{-1} \, b \,  \varphi'(Z_+^0) \varphi''(Z_-^0)\,, \\
	B = - a \varphi'(Z_+^0) -  b^{-1} \varphi'(Z_-^0) + 2 g_0 \, (ab^{-1} - a^{-1}b) \, \varphi'(Z_+^0) \varphi'(Z_-^0)\,.
\end{gather}
Using eqs. (\ref{eq:zpzmg-zero}) and (\ref{eq:alpha}), we get
\begin{equation}
	\det J = (Z_+^0 \, b \, \varphi'(Z_+^0) + g_0^{-1} \, Z_-^0 \, (1 - g_0 a \varphi'(Z_+^0)))
			\begin{vmatrix}
				1 & - \alpha\\
				A_+ & A_-
			\end{vmatrix}
		< 0 \,,
\end{equation}
since clearly $A_\pm<0$ and $\alpha>0$ by Remark
\ref{rem:c-alpha}. This proves the claim.
\end{proof}

We can now establish the following result for the single site magnetization on the spine.

\begin{thm}\label{thm:nomag-spine}
Under the assumptions of Theorem \ref{thm:wl}, the probability\linebreak $\mu^{(\beta,h)}(\set{s_v = +1 })$ is a smooth
function of $\beta, h$  for any spine vertex $v$. In particular, there is 
no spontaneous magnetization in the sense that 
\begin{equation}
	\lim_{h \to 0} \mu^{(\beta,h)}(\set{ s_v = + 1 }) = \frac{1}{2}\,.
\label{eq:mag-spine}
\end{equation}
\end{thm}

\begin{proof}
For the root vertex $r$, we have by eq. (\ref{eq:wl}) that
\begin{equation}
	\mu^{(\beta,h)}(\set{ s(r) = + 1 }) = \frac{\alpha(\beta,h)}{1 + \alpha(\beta,h)}\,,
\label{eq:mag-root}
\end{equation}
where $\alpha(\beta,h)$ is given by (\ref{eq:alpha}) and is a smooth
function of $\beta, h$ by Proposition \ref{prop:zpm-smooth}. Hence, to
verify (\ref{eq:mag-spine}) for $v=r$ it suffices to note that
$\alpha(\beta,0) = 1$, since $a=b^{-1}$ and $Z_+^0 = Z_-^0$ for $h=0$.

Now, assume $v=u_N$ is at distance $N$ from the root, and define
\begin{equation}
	p_{ij} = \mu_i(\set{s_v = j }) \, \frac{\alpha^{\frac{1+i}{2}}}{1 + \alpha}\,,
\label{eq:prob-vpm}
\end{equation} 
for $i$, $j \in \set{ \pm 1 }$, where we use $\pm 1$ and $\pm$ interchangeably. From eq. (\ref{eq:prob-spine}) follows that
\begin{equation}
	\begin{split}
	\mu_{s_0}(\set{s_v = s_N}) & = \sum_{\substack{k'_i,k''_i \geq 0\\ s_1,\ldots,s_{N-1}}} \rho^{s_0}_{k'_1,\ldots,k'_N,k''_1,\ldots,k''_N} (s_0,\ldots,s_N) \\
	& = \sum_{s_1,\ldots,s_{N-1}} \prod_{i=1}^N g_0 [\Phi'(Z_+^0,Z_-^0)]_{s_{i-1}s_i} \, \alpha^{\frac{s_N - s_0}{2}} \\
	& = [(g_0 \, \Phi'(Z_+^0,Z_-^0))^N]_{s_0 s_N} \, \alpha^{\frac{s_N - s_0}{2}},
	\end{split}
\end{equation}
where we have used that the matrix elements of $\Phi'(Z_+^0,Z_-^0)$ are
given by
\begin{equation}
	[\Phi'(Z_+^0,Z_-^0)]_{s_{i-1}s_i} = e^{\beta s_{i-i} s_i + h s_i} \varphi'(Z_{s_i}^0)\,.
\end{equation}
Hence, substituting into (\ref{eq:prob-vpm}) we have
\begin{equation}
	p_{ij} = \left[ (g_0 \Phi'(Z_+^0,Z_-^0))^N \right]_{ij} \, \frac{\alpha^{\frac{1+j}{2}}}{1 + \alpha}.
\label{eq:prob-vpm-2}
\end{equation}
By Proposition \ref{prop:zpm-smooth}, all factors on the RHS of (\ref{eq:prob-vpm-2}) are
smooth functions of $\beta,h$, and by (\ref{eq:wl}) we have
\begin{equation}
	\mu^{(\beta,h)}(\set{s_v = j }) = p_{+j} + p_{-j}\,.
\end{equation}
Eq. (\ref{eq:mag-spine}) is now obtained from (\ref{eq:prob-vpm-2}) by
noting again that for $h=0$ we have $\alpha=1$ and hence $c_1 = c_2$ (see Remark \ref{rem:c-alpha}), which gives
\begin{equation}
	\begin{split}
		p_{+j} + p_{-j} & = 
		\left[ (1 \quad 1)(g_0 \, \Phi'(Z^0,Z^0))^N	\right]_j \, \frac{1}{2}\\
		& = 	(1 \quad 1)_j \, \frac{1}{2} = \frac{1}{2}\,.
	\end{split}
\end{equation} 
\end{proof}

The preceding proof together with (\ref{eq:prob-spine}) shows that the
distribution of spin variables $s_0,\ldots,s_N$ on the spine can be
written in the form 
\begin{equation}
	\rho(s_0,\ldots,s_N) = e^{-H'_N(s_0,\ldots,s_N)} \left( g_0^2 \, \varphi'(Z_+^0) \varphi'(Z_-^0) \right)^{N/2} \, \frac{\sqrt{\alpha}}{1 + \alpha}
\end{equation}
where
\begin{equation}
	H'_N(s_0,\ldots,s_N) = - \beta \sum_{i=1}^N s_{i-1} s_i - h' \sum_{i=1}^N s_i - \frac{s_N}{2} \log\alpha
\label{eq:hamil-spine}
\end{equation}
and 
\begin{equation}\label{hprime}
	h' = h + \frac{1}{2} \, \ln \frac{\varphi'(Z_+^0)}{\varphi'(Z_-^0)}\,.
\end{equation}
Since $\rho(s_0,\ldots,s_N)$ is normalized, the expectation value
w.r.t. $\mu$ of a function $f(s_0,\ldots,s_{N-1})$ hence coincides
with the expectation value w.\,r.\,t. the Gibbs measure of the Ising
chain on $[0,N]$, with Hamiltonian given by (\ref{eq:hamil-spine}) and
\eqref{hprime}. 
In particular, we have that the mean magnetization on the spine
vanishes in the absence of an external magnetic field, since $h'$ is a
smooth function of $h$, by Proposition \ref{prop:zpm-smooth}, and
vanishes for $h=0$ (see e.g. \cite{Baxter} for details about the 1d Ising model). 

We state this result as follows.

\begin{cor} Under the assumptions of Theorem \ref{thm:wl}, the mean magnetization on the spine vanishes as $h \to 0$, i.e.
\begin{equation}
	\lim_{h \rightarrow 0} \lim_{N \rightarrow \infty} \left< \frac{s_0 + \cdots + s_{N-1} }{N} \right>_{\beta,h} = 0.
\end{equation}
\end{cor}

\subsection{Mean magnetization}
\label{sec:mean-mag}

For the mean magnetization on the full infinite tree, defined in Sec. \ref{sec:sm-mag}, we have the following result,
which requires some additional estimates in combination with Proposition \ref{prop:zpm-smooth}.
\begin{thm}
\label{thm:no-mean-mag}
Under the assumptions of Theorem \ref{thm:wl}, the mean magnetization vanishes for
$h \rightarrow 0$, i.e.
\begin{equation}
	\lim_{h \rightarrow 0} M(\beta,h) = 0\,, \quad \beta \in
        \mathbb{R}\,,
\end{equation}
where $M(\beta,h)$ is defined by (\ref{mag1})-(\ref{mag2}).
\end{thm}

\begin{proof}
Consider the measure $\nu_\pm$ given by (\ref{eq:prob-branch}) and, for a given finite branch $T$, let $S_R(T)$ denote the sum of spins at distance $R$ from the root of $T$. Setting
\begin{equation}
	m_R^\pm =  Z_\pm^0 \left< S_R \right>_{\nu_\pm} 
\end{equation}  
it follows, by decomposing $T$ according to the spin and the degree of the vertex closest to the root, that  
\begin{equation}
	\begin{cases}
		m_R^+ = g_0 \, \left(a \, \varphi'(Z_+^0) \, m_{R-1}^+ + a^{-1} \, \varphi'(Z_-^0) \, m_{R-1}^-\right)\\
		m_R^- = g_0 \, \left(b \, \varphi'(Z_+^0) \, m_{R-1}^+ + b^{-1} \, \varphi'(Z_-^0) \, m_{R-1}^-\right),
	\end{cases}
\label{eq:rec-mag}
\end{equation}
for $R\geq 2$, and $m_0^\pm = \pm Z_\pm^0$.
In matrix notation these recursion relations read
\begin{equation}
	m_R = g_0 \Phi'_0 \, m_{R-1}, 
\end{equation}
which, upon  multiplication by the left eigenvector $c$ of $g_0 \Phi'_0$, leads to
\begin{equation}
	c \, m_R = g_0 \, c \, \Phi'_0 \, m_{R-1} = c \, m_{R-1},
\end{equation}
and hence 
\begin{equation}
		c_1 \, m_R^+ + c_2 \, m_R^- = c_1 \, m_1^+ + c_2 \, m_1^-, \qquad R \geq 1\,. 
\label{eq:cmr}
\end{equation}

Now, fix $N \geq 1$ and let $U_{R,N}$ denote the sum of all spins at
distance $R \geq 1$ from the $N$'th spine vertex $u_N$ in the branches
attached to $u_N$. The conditional expectation of $U_{R,N}$, given
$s_0, s_1, \ldots, s_N$, then only depends on $s_N$, and its value is
obtained from Corollary \ref{cor:measures-description} by summing over
$k'_N, k''_N \geq 0$, which yields
\begin{equation}
\begin{split}
	\left( \sum_{k = 0}^\infty (Z_{s_N}^0)^k (k+1) p_{k+1} \right)^{-1} & \sum_{k = 0}^\infty (Z_{s_N}^0)^{k-1} p_{k+1} k (k+1) \, m_R^{s_N} \\
	& = \varphi'(Z_{s_N}^0)^{-1} \varphi''(Z_{s_N}^0) \, m_R^{s_N} \equiv d_R^{s_N}\,.
\end{split}
\label{eq:drpm}
\end{equation}
Using the matrix representation (\ref{eq:prob-vpm-2}) for $p_{ij}$, this gives
\begin{equation}\label{URN}
	\left< U_{R,N} \right>_{\beta,h} = \frac{1}{1 + \alpha} \, (1 \quad 1) \, (g_0 \Phi'(Z_+^0,Z_-^0))^N
	\begin{pmatrix}
	\alpha \, d_R^+ \\
	d_R^-
	\end{pmatrix}\,.
\end{equation}

As pointed out in Remark \ref{rem:c-alpha}, the matrix $g_0 \Phi'_0$
has a second left eigenvalue $\lambda$ such that $|\lambda|<1$. Let
$(e_1, e_2)$ be a smooth choice of eigenvectors
corresponding to $\lambda$ as a function of $(\beta, h)$, e.g. 
\begin{equation}
	\begin{pmatrix}
	e_1 \\ 
	e_2
	\end{pmatrix}
	=
	\begin{pmatrix}
	g_0 b \varphi'(Z_+^0)\\
	g_0 b^{-1} \varphi'(Z_-^0)-1
	\end{pmatrix}\,,
\end{equation}
and write
\begin{equation}
	\begin{pmatrix}
	1 \\
	1
	\end{pmatrix}
	= A
	\begin{pmatrix}
	c_1 \\
	c_2
	\end{pmatrix}
	+
	B
	\begin{pmatrix}
	e_1 \\
	e_2
	\end{pmatrix}\,.
\end{equation}
From \eqref{URN} we then have
\begin{equation}
\begin{split}
	\left< U_{R,N} \right>_{\beta,h} & = \frac{1}{1 + \alpha} \,
        \left( A (c_1 \quad c_2) + B \lambda^N (e_1 \quad e_2)\right) 
	\begin{pmatrix}
	\alpha \, d_R^+\\
	d_R^-
	\end{pmatrix} \\
	&= \frac{A}{1 + \alpha} ( c_1 \alpha d_R^+ + c_2 d_R^-) + \lambda^N \frac{B}{1 + \alpha} ( e_1 \alpha d_R^+ + e_2 d_R^- )\,,
\end{split}
\label{eq:URN-calc}
\end{equation}
and from (the proof of) Theorem \ref{thm:nomag-spine} it follows that $A
\rightarrow \tilde{c}^{-1}$ and $B \rightarrow 0$ for $h \rightarrow
0$, where $\tilde{c} = c_1(\beta,0) = c_2(\beta,0)$.
 
Next, note that  $|d_R^\pm|, R\geq 1,$ are bounded by a constant $C_1=C_1(\beta,h)$ as a
consequence of  (\ref{eq:fR-bound}), and that
\begin{equation}
\begin{split}
	\left< M_R(\beta,h) \right>_{\beta,h} & \leq
        \left<|B_R|\right>_{\beta,h}^{-1}\sum_{R',N \leq R}\left| \left<
            U_{R',N} \right>_{\beta,h} \right|\\
&\leq C_2 R^{-2}\sum_{R',N \leq R}\left| \left<
            U_{R',N} \right>_{\beta,h} \right|
\end{split}
\end{equation}
 for some constant $C_2=C_2(\beta,h)$ by \eqref{volball}. It now
 follows from \eqref{eq:URN-calc} that
\begin{equation}
	\left| \left< M_R(\beta,h) \right>_{\beta,h} \right| \leq  \frac{A\,C_2}{R(1 + \alpha)} \sum_{R'=1}^R ( c_1 \alpha d_R^+ + c_2 d_R^- ) + R^{-1} \, B \, C_1\,C_2 \, \max\{ e_1,e_2\}.
\end{equation}
Obviously, the second term on the RHS vanishes in the limit
$R\to\infty$. Rewriting the summand in the first term on the RHS as
\begin{equation}
\begin{split}
	c_1 \alpha d_R^+ + c_2 d_R^- & = c_1 \, \alpha \, m_R^+ \varphi'(Z_+^0)^{-1} \varphi''(Z_+^0) + c_2\,  m_R^- \varphi'(Z_-^0)^{-1} \varphi''(Z_-^0) \\
	& = (c_1 m_R^+ + c_2 m_R^-) \, \varphi'(Z^0)^{-1} \varphi''(Z^0) \\
	& +  c_1 m_R^+ \left[ \alpha \varphi'(Z_+^0)^{-1} \varphi''(Z_+^0) - \varphi'(Z^0)^{-1} \varphi''(Z^0) \right] \\
	& +  c_2 m_R^- \left[ \varphi'(Z_-^0)^{-1} \varphi''(Z_-^0) - \varphi'(Z^0)^{-1} \varphi''(Z^0) \right]\,,
\end{split}
\end{equation}
we see the last two terms in this expression tend to $0$ uniformly in
$R$ as $h\to 0$ by continuity of $Z_\pm^0$, $g_0$
and boundedness of $|m_R^\pm|$, and the same holds for the first term
as a consequence of (\ref{eq:cmr}) 
and continuity of $c_1$, $c_2$, $Z_\pm^0$ and
$g_0$. In conclusion, given $\epsilon>0$ there exists $\delta>0$ such that 
\begin{equation}
	\left| \left< M_R(\beta,h) \right>_\mu \right| \leq \epsilon \frac{A\,C_2}{1 + \alpha}  + C' \, R^{-1}\,,
\end{equation}
if $|h|<\delta$, where $C'$ is a constant. This completes the proof of
the theorem.
\end{proof}

\begin{rem} A natural alternative to the mean magnetization
as defined by \eqref{mag1}-\eqref{mag2} is the quantity 
\begin{equation}
\bar M(\beta,h) = \limsup_{R\to\infty} \bar M_R(\beta,h)\,,
\end{equation}
where
\begin{equation}
	\bar{M}_R(\beta,h) =  \left<|B_R(\tau)|^{-1}\sum_{v \in B_R(\tau)}
          s_v \right>_{\beta,h}\,.
\end{equation}
It is natural to conjecture that $\lim_{h\to 0}\bar M(\beta,h) = 0$ holds for generic Ising trees.
\end{rem}

\section{Two-point function}

With the notation introduced in the previous section, we define the connected \emph{two-point function} as
\begin{equation}
	G_{\beta,h}(R) = \frac{1}{\langle |D_R| \rangle_{\beta,h}} \left( \langle V_R s_0 \rangle_{\beta,h} - \langle V_R \rangle_{\beta,h} \langle s_0 \rangle_{\beta,h} \right),
\label{eq:}
\end{equation} 
where $V_R$ denotes the sum of the spins at distance $R$ from the root, that is
\begin{equation}
	V_R = \sum_{N+R'=R} U_{R',N}.
\end{equation}
Applying the same techniques as above we find
\begin{equation}
	G_{\beta,h}(R) = \frac{2 \alpha}{(1+\alpha)^2} \, \frac{1}{\langle|D_R|\rangle_{\beta,h}}\sum_{N + R'=R} \lambda^N (d_{R'}^+ - d_{R'}^-),
\end{equation}
where $d_{R'}^\pm$ are defined in (\ref{eq:drpm}). The explicit expression for $d_{R}^\pm$ can be obtained from the system (\ref{eq:rec-mag}), which gives
\begin{equation}
	d_{R}^\pm  = B_\pm + C_\pm \, \lambda^{R}, \qquad R \geq 1\,,
\end{equation}
with
\begin{equation}
	B_+ = -\frac{\varphi''(Z^0_+)}{\varphi'(Z^0_+)}\frac{(c_1m_1^+
          + c_2 m_1^-)e_2}{c_2 e_1 - c_1 e_2}, \qquad C_+ =
        \frac{\varphi''(Z^0_+)}{\varphi'(Z^0_+)} \frac{(e_1 m_1^+ +
          e_2 m_1^-) c_2}{c_2 e_1 - c_1 e_2}  
\end{equation}
and
\begin{equation}
	B_- =  \frac{\varphi''(Z^0_-)}{\varphi'(Z^0_-)} \frac{(c_1 m_1^+ + c_2 m_1^-)e_1}{c_2 e_1 - c_1 e_2}, \qquad C_- = - \frac{\varphi''(Z^0_-)}{\varphi'(Z^0_-)} \frac{(e_1 m_1^+ + e_2 m_1^-) c_1}{c_2 e_1 - c_1 e_2}.
\end{equation}
Hence, using $d_0^\pm = \pm 1$, we have 
\begin{equation}
	G_{\beta,h}(R) = \frac{2 \alpha}{(1+\alpha)^2} \frac{1}{\langle |D_R| \rangle_{\beta,h}} \left[ (B_+ - B_-) \frac{\lambda - \lambda^R}{1-\lambda} + (C_+ - C_-) (R-1) \lambda^R + 2 \lambda^R  \right]
\label{eq:2pf}
\end{equation}
Restricting attention to $\beta, h\geq 0$ standard correlation inequalities
imply that $Z_-(\beta,h)\leq Z_+(\beta,h)$ and $G_{\beta,h}(R)\geq 0$, and consequently $B_+-B_-\geq
0$. For $h=0$ we have $m_1^- =- m_1^+,\,c_1=c_2, e_1=-e_2$
and hence $B_+=B_-=0$, and
$\langle|D_R|\rangle_{\beta,0}$ can be calculated
explicitly. Eq. \eqref{eq:2pf} then reduces to
\begin{equation}
	G_{\beta,0}(R) =  \lambda^R = (\tanh \beta)^R,
\end{equation} 
a result that is easily obtainable directly by computing the
two-point function on finite trees. For $h>0$ we generally have
$B_+\neq B_-$ in which case 
$$
G_{\beta,h}(R) \sim \frac{1}{R}
$$ for $R$ large, since 
$$
{\langle |D_R| \rangle_{\beta,h}} \sim R
$$
as a consequence of eq. (\ref{eq:fR-bound}) and Corollary
\ref{cor:measures-description}. Thus the decay of $G_{\beta,h}(R)$ in
this case is entirely determined by the behaviour of ${\langle |D_R|
  \rangle_{\beta,h}}$. It is conceivable that the alternative two-point function
$\bar G_{\beta,h}(R)$ defined by 
\begin{equation}
	\bar{G}_{\beta,h}(R) = \left< \frac{V_R}{|D_R|} s_0 \right>_{\beta,h} - \left< \frac{V_R}{|D_R|} \right>_{\beta,h} \left< s_0 \right>_{\beta,h}.
\end{equation}
decays exponentially, but our methods do not suffice to prove it. For
$h=0$ it is easy to see that $G_{\beta,0}(R)=\bar G_{\beta,0}(R)$.

\section{Conclusions}\label{sec:concl}
The statistical mechanical models on random graphs considered in this
paper possess two simplifying features, beyond being Ising models, the
first being that the graphs are restricted to be trees and the second
that they are generic, in the sense of \eqref{eq:genass}. Relaxing the
latter condition might be a way of producing models with different
magnetization properties from the ones considered here. Infinite 
non-generic trees having a single vertex of infinite degree have been
investigated in  \cite{JS09,JS11}, but it is unclear whether a
non-trivial coupling to the Ising model is possible. A different
question is whether validity of the genericity condition
\eqref{eq:genass} for $h=0$ implies its validity for all $h\in \mathbb{R}$. 
The arguments presented in Section \ref{sec:part-func} only show that the
domain of genericity in the $(\beta,h)$-plane is an open subset
containing the $\beta$-axis, and thus leaves open the possibility of
a transition to non-generic behavior at the boundary of this set.

Coupling the Ising model to other ensembles of infinite graphs
represents a natural object of future study. In particular, models
of planar graphs may be tractable. The so-called uniform infinite causal triangulations
of the plane are known to be closely related to plane trees
\cite{DJW3,KY}, and a quenched version of this model coupled to the
Ising model without external field has been considered in \cite{KY},
and found to have a phase transition. Analysis of the non-quenched version,
analogous to the models considered in the present paper, seem to
require developing new techniques. Surely, this is also the case for other 
planar graph models such as the uniform infinite planar triangulation \cite{AS}
or quadrangulation \cite{CD}.

\ack

This work is supported by the Danish Agency for Science, 
Technology and Innovation through the Geometry and Mathematical
Physics School (GEOMAPS), by the NordForsk researcher network in Random Geometry,
 and by the Danish National Research Foundation (DNRF) through
the Centre for Symmetry and Deformation.

\section*{References}

\bibliography{bibtrees}{}
\bibliographystyle{habbrv}

\end{document}